\documentclass{article}
\usepackage{graphicx} % Required for inserting images
\usepackage{amsmath}
\usepackage{amssymb}
\usepackage{amsthm}
\usepackage{hyperref}
\usepackage{comment}
\usepackage{xcolor}
\usepackage{subcaption}
\usepackage{authblk}

\newcommand{\R}{\ensuremath{\mathbb{R}}}

\newcommand{\dive}{\mathrm{div}}
\newcommand{\tr}{\mathrm{tr}}

\usepackage[style=numeric,backend=biber,sorting=ynt]{biblatex}

\newtheorem{prop}{Proposition}
\newtheorem{remark}{Remark}

\newtheorem{proposition}{Proposition}
\newtheorem{lemma}{Lemma}

\title{Covariance-Adaptive Bouncy Particle Samplers via Split Lagrangian Dynamics}
\author[1]{A. Chevallier}
\author[ ]{E. Raab}
\affil[1]{Université de Strasbourg}
\begin{document}

\maketitle
\begin{abstract}
    Piecewise Deterministic Markov Processes (PDMPs) provide a powerful framework for continuous-time Monte Carlo, with the Bouncy Particle Sampler (BPS) as a prominent example. Recent advances through the Metropolised PDMP framework allow local adaptivity in step size and effective path length, the latter acting as a refreshment rate. However, current PDMP samplers cannot adapt to local changes in the covariance structure of the target distribution. 
    
    We extend BPS by introducing a position-dependent velocity distribution that varies with the local covariance structure of the target. Building on ideas from Riemannian Manifold Hamiltonian Monte Carlo and its velocity-based variant, Lagrangian Dynamical Monte Carlo, we construct a PDMP for which changes in the metric trigger additional velocity update events. Using a metric derived from the target Hessian, the resulting algorithm adapts to the local covariance structure.  
    Through a series of controlled experiments, we provide practical guidance on when the proposed covariance-adaptive BPS should be preferred over standard PDMP algorithms.
\end{abstract}

\section{Introduction}

Piecewise Deterministic Markov Processes (PDMPs) have recently emerged as a promising class 
of continuous-time Monte Carlo algorithms, offering non-reversible dynamics, rejection-free 
proposals, and competitive performance in high-dimensional settings \cite{fearnhead2018piecewise}. Among these, the 
\emph{Bouncy Particle Sampler} (BPS) \cite{bouchard2018bouncy} has become one of the most 
widely used examples, combining deterministic linear trajectories with stochastic bounce and 
refreshment events.

Recent work has extended BPS through the framework of Metropolised PDMPs \cite{chevallier2025practicalpdmpsamplingmetropolis}, in which numerical approximations of the event rate are combined with a Metropolis-Hastings correction. This approach removes the need for manually derived upper bounds on event rates, which had to be tailored to each target. This requirement was a serious limitation, since it discouraged the design of PDMPs with more sophisticated rates. At the same time, Metropolised PDMPs make it possible to introduce local adaptivity, most notably in the step size of the numerical approximation and the effective path length of the process, the latter acting as a mechanism for controlling the refreshment rate. In addition, non-local adaptation has been proposed for the covariance matrix \cite{bertazzi2022adaptive}. While these methods represent a significant advance, current PDMP-based samplers still cannot adapt to local changes in the covariance structure of the target distribution.

The goal of this paper is to address this limitation by extending BPS to incorporate 
\emph{position-dependent velocity distributions}, informed by a Riemannian metric on the position space. Our construction draws inspiration from Riemannian Manifold Hamiltonian Monte Carlo 
(RMHMC) \cite{girolami2011riemann}, and in particular from the Lagrangian Dynamical Monte Carlo formulation \cite{Lan_2015}, which replaces the momentum formulation of RMHMC with a velocity-based parametrization. Translating these ideas into the PDMP framework yields a new variant of BPS that adapts naturally to the local covariance structure: when the metric changes with position, additional events are introduced to update the velocity accordingly. 

In this work, we employ a metric based on the Hessian of the log-density of the target distribution, following \cite{betancourt2013general}. This metric captures local covariance structure of the target and adapts the dynamics to anisotropic targets.

The paper is organised as follows. Section~\ref{sec:previous-work} reviews background on PDMP sampling and Lagrangian Dynamical Monte Carlo. Section~\ref{sec:new-algo} introduces a PDMP based on split Lagrangian dynamics and our extension of BPS with position-dependent velocities. Section~\ref{sec:simulation} outlines the numerical simulation of these processes. Section~\ref{sec:experiments} reports controlled experiments designed to indicate when the proposed algorithm should be preferred to standard BPS. Section~\ref{sec:discussion} discusses the strengths and limitations of the method and concludes with directions for future research.

%%%%%%%%%%%%%%%%%%%%%%%%%%%%%%%%%%%%%%%%%%%%%%%%%%%%%%%%%%%%%%%%%%%%%%%%%%%
%%%%%%%%%%%%%%%%%%%%%%%%%%%%%%%%%%%%%%%%%%%%%%%%%%%%%%%%%%%%%%%%%%%%%%%%%%%
\section{Previous work}
\label{sec:previous-work}

%%%%%%%%%%%%%%%%%%%%%%%%%%%%%%%%%%%%%%%%%%%%%%%%%%%%%%%%%%%%%%%%%%%%%%%%%%%
\subsection{Piecewise deterministic Markov processes samplers}
A \emph{piecewise deterministic Markov process} (PDMP) is a continuous-time Markov process 
$(X_t)_{t \geq 0}$ whose dynamics are governed by deterministic flows punctuated by random 
events \cite{davis1984piecewise}. Between events, the state evolves deterministically according to an ordinary differential 
equation (ODE):
\[
\frac{d}{dt} X_t = \Phi(X_t),
\]
where $\Phi$ is a vector field. Event times are generated according to an inhomogeneous 
Poisson process with rate $\lambda(X_t)$. In other words, the probability of having an event in the time interval $[t,t+\mathrm{d}t]$ is $\lambda(X_t)\,\mathrm{d}t + o(\mathrm{d}t)$. At each event time, the process undergoes a 
stochastic update
\[
X_t \sim Q(\cdot \mid X_t),
\]
where $Q$ is a transition kernel.

In PDMP-based Monte Carlo, the triplet $(\Phi, \lambda, Q)$ is chosen so that a target probability measure $\mu$ is invariant \cite{fearnhead2018piecewise}. 
Let $\rho(x,t)$ denote the probability density of $X_t$ at position $x$ and time $t$. 
Between events, the density flows deterministically along $\Phi$, satisfying the Liouville equation
\[
\frac{\partial \rho}{\partial t} + \mathrm{div}(\rho \, \Phi) = 0.
\]
Events occur at rate $\lambda(x)$ and redistribute probability according to $Q$, contributing
\[
\frac{\partial \rho}{\partial t}\Big|_\text{events} = \int \lambda(x') \, Q(x \mid x') \, \rho(x',t)\,dx' - \lambda(x) \, \rho(x,t).
\]
Combining these terms gives the Kolmogorov forward equation (or Fokker--Planck equation) for the PDMP:
\[
\frac{\partial \rho(x,t)}{\partial t} = - \mathrm{div}(\rho(x,t) \, \Phi(x)) + \int \lambda(x') \, Q(x \mid x') \, \rho(x',t)\,dx' - \lambda(x) \, \rho(x,t).
\]
Stationarity is thus achieved when $\rho(x,t) = \mu(x)$ satisfies
\[
\mathrm{div}(\mu(x) \, \Phi(x)) = \int \lambda(x') \, Q(x \mid x') \, \mu(x') \, dx' - \lambda(x) \, \mu(x),
\]
which expresses "mass conservation": any change of density due to the deterministic flow is offset by stochastic events to and from the state.

A common special case occurs when events correspond to a deterministic, bijective, volume-preserving map $F$:
\[
X_t \mapsto F(X_t), \quad \text{with } |\det DF(x)| = 1.
\]
In this case, the jump term in the Kolmogorov forward equation simplifies, and the density evolution becomes
\[
\frac{\partial \rho(x,t)}{\partial t} = - \mathrm{div}(\rho(x,t) \, \Phi(x)) + \lambda(F^{-1}(x)) \, \rho(F^{-1}(x),t) - \lambda(x) \, \rho(x,t),
\]
where the gain term accounts for mass arriving from the unique pre-image under $F$. 
Consequently, stationarity is achieved when $\rho(x,t) = \mu(x)$ satisfies
\begin{equation}
    \label{eq:pdmp-inv}
    \mathrm{div}(\mu(x) \, \Phi(x)) = \lambda(x) \, \mu(F^{-1}(x)) - \lambda(x) \, \mu(x),    
\end{equation}
which expresses conservation of mass under deterministic, bijective, volume-preserving jumps.

\begin{remark}
    While the mass-conservation perspective provides an intuitive understanding of invariance, a fully rigorous treatment of PDMPs requires the formalism of the extended infinitesimal generator \cite{davis1984piecewise, durmus2021piecewise,chevallier2024pdmp}.
\end{remark}

\paragraph{Bouncy Particle Sampler}\mbox{}\\

The \emph{Bouncy Particle Sampler} (BPS) \cite{bouchard2018bouncy} is a PDMP-based Markov chain 
Monte Carlo (MCMC) method for sampling from a target density 
$\pi$. The process is defined on the extended state space 
\[
(x,v) \in \mathbb{R}^d \times \mathbb{R}^d,
\]
and admits the stationary density
\[
\mu(x,v) = \pi(x)\,p_v(v),
\]
where $p_v(v)$ is a chosen distribution for the velocity. Two common choices are:
\begin{itemize}
    \item the \emph{uniform distribution on the unit sphere}, 
    $p_v(v) = \text{Unif}(\{v \in \mathbb{R}^d : \|v\|=1\})$;
    \item the \emph{standard Gaussian distribution}, $p_v(v) = \mathcal{N}(0,I_d)$.
\end{itemize}
The deterministic dynamics are governed by the vector field
\[
\Phi(x,v) = (v, 0),
\]
so that between events the system evolves according to
\[
\frac{d}{dt} X_t = V_t, 
\qquad 
\frac{d}{dt} V_t = 0,
\]
i.e. linear motion.
Event times are generated according to the inhomogeneous Poisson rate
\[
\lambda(x,v) = [-\langle v, \nabla \log \pi\rangle]^+.
\]
At such an event, the velocity is updated deterministically by reflection:
\[
F: V \mapsto V - 2 \frac{\langle V, \nabla \log \pi \rangle}{\|\nabla \log \pi \|^2} \, \nabla \log \pi,
\]
which corresponds to a ``bounce'' against the energy gradient. 

One readily verifies that equation \eqref{eq:pdmp-inv} is satisfied, hence $\mu$ is stationary. This follows from directly from
\begin{enumerate}
    \item $\mathrm{div}(\mu\Phi)=\mu\,v\cdot\nabla \log \pi$,
    \item The reflection $F$ reverses the sign of $\langle v,\nabla \log \pi\rangle$ hence $\lambda(F(x,v))-\lambda(x,v)= v\cdot\nabla \log \pi$ ,
    \item $F$ preserves volumes and $p_v$ is invariant under $F$, so $\mu(F(x,v))=\mu(x,v)$. 
\end{enumerate}

\begin{remark}
    To ensure ergodicity, additional \emph{refreshment events} are introduced at a constant 
rate, where $V$ is resampled from $p_v$.
\end{remark}

%%%%%%%%%%%%%%%%%%%%%%%%%%%%%%%%%%%%%%%%%%%%%%%%%%%%%%%%%%%%%%%%%%%%%%%%%%%
\subsection{Metropolised PDMP samplers}
\label{subsec:metropolised-PDMP}

Exact simulation of PDMP samplers typically requires manually derived upper bounds on the event rates 
to implement the standard thinning procedure. For complex rates, 
this approach is often infeasible. To address this limitation, we adopt the Metropolised PDMP framework 
\cite{chevallier2025practicalpdmpsamplingmetropolis}. 
For clarity, we present the framework in a simplified setting with:
\begin{enumerate}
    \item a target measure with density $\mu$,
    \item two distinct event rates, $\lambda_1$ and $\lambda_2$,
    \item deterministic transition kernels $F_1(z)$ and $F_2(z)$ that preserve volume and that satisfy\ldots
    \item $\mu(F_i(z)) = \mu(z)$ for all $z$.
\end{enumerate}

The Metropolised PDMP framework interprets an approximately simulated PDMP trajectory as a Metropolis--Hastings proposal on path space. It assumes that the deterministic flow of the PDMP has a known exact analytical solution, while the event rates are approximated numerically. Specifically, the approximate rates $\tilde\lambda_1$ and $\tilde\lambda_2$ are computed along the trajectory, for instance by replacing the true rates with piecewise-constant approximation.

To apply a Metropolis--Hastings correction, it is necessary to define the probability density of a given path, which in turn requires a description of path space. Here, a path is specified by the initial state $z_0$, the event times $t_1,\dots,t_k$, and the type of each event, indicating whether it arises from the first or second rate. The path space over the time interval $[0,T]$ is then given by
\[
    W = E \times \bigcup_{k \in \mathbb{N}} \left(\, ]0,T[ \times \{1,2\} \,\right)^k,
\]
where $E$ denotes the state space of the PDMP. The space $W$ is endowed with the Lebesgue measure on $]0,T[^k$ on each element of the union, allowing us to define densities. The probability density of a given path $\tilde w = (\tilde z_0,t_1,a_1,...,t_k,a_k)$, where $a_i \in \{1,2\}$ is then
\begin{equation}
    \label{eq:proba-path}
    \tilde p(\tilde w \mid z_0)
    = \Bigg[ \prod_{k=1}^k 
        q_i \tilde\lambda(\tilde z_{t_i^-}) 
        \exp\!\Big( -\!\!\int_{t_{i-1}}^{t_i} \tilde\lambda(\tilde z_s)\,ds \Big)  \Bigg]
        \exp\!\Big( -\!\!\int_{t_k}^{T} \tilde\lambda(\tilde z_s)\,ds \Big),
\end{equation}
where $\tilde\lambda = \tilde \lambda_1 + \tilde \lambda_2$ and $q_i = \frac{\mathbf{1}_{a_i = 1}\tilde\lambda_1(\tilde z_{t_i^-}) + \mathbf{1}_{a_i = 2}\tilde\lambda_2(\tilde z_{t_i^-})}{\tilde \lambda(\tilde z_{t_i^-})}$.

A final ingredient is the use of the time-reversed process. If a path leads from $\tilde z_0$ to $\tilde z_T$, detailed balance requires the ability to construct a corresponding path from $\tilde z_T$ back to $\tilde z_0$. This is naturally achieved by considering the time-reversal of the PDMP, which is itself a PDMP \cite{lopker2013time}. In our simplified setting, the paths of the reversed process are defined on the same path space $W$.
The application $R:W \rightarrow W$ that reverses a path is defined as 
\[
    R: (\tilde z_0,t_1,a_1,...,t_k,a_k) \mapsto (\tilde z_T,T - t_k,a_k,...,T - t_1,a_1).
\]
In this simplified setting, the reverse process is the PDMP that:
\begin{enumerate}
    \item follows the deterministic flow "backwards in time",
    \item uses event rates $\lambda_i^r(z) = \lambda_i(F_i^{-1}(z))$,
    \item and applies the jump transition kernel $z \mapsto F_i^{-1}(z)$,
\end{enumerate}
and can therefore be simulated just as easily as the forward process.

Starting from a state $z_0$ and a time-direction $\gamma \in \{\pm 1\}$, one generates a piecewise-deterministic trajectory 
$\tilde w = (z_t)_{t \in [0,T]}$ either forward or backward in time. Then the end point $\tilde z_T$ is accepted with probability 
\begin{equation}
    \label{eq:acceptance-proba}
    \alpha(\tilde w) =
    \begin{cases}
    1 \wedge \dfrac{\mu(\tilde w_T)\,\tilde{p}^r(R(\tilde w)\mid \tilde w_T)}
    {\mu(\tilde z_0)\,\tilde{p}(\tilde w \mid \tilde z_0)\,\psi(R(\tilde w))}, & \gamma = 1, \\[1.5em]
    1 \wedge \dfrac{\mu(\tilde w_T^r)\,\tilde{p}(R(\tilde w^r)\mid \tilde w_T^r)\,\psi(\tilde w^r)}
    {\mu(\tilde z_0)\,\tilde{p}^r(\tilde w^r \mid \tilde z_0)}, & \gamma = -1,
    \end{cases}
\end{equation}
where $\psi$ is the volume change associated to $R$, that is for any $f: W \mapsto W$,
\[
    \int_W f(R(w)) dw = \int_W f(w) \psi(w) dw.
\]

\begin{remark}[Change of volume]
    \label{rmk:vol-change}
    In our setting, the volume change induced by $R$ reduces to the volume change of the mapping that takes $\tilde z_0$ to $\tilde z_T$, given fixed event times and types. This change can be computed from the local volume expansion along the deterministic flow, since the event transition kernels are assumed to preserve volume.
\end{remark}

\subsection{Split flow PDMP samplers}
\label{subsec:flow-splitting}

This section describes a splitting technique introduced in \cite{power_2020}.
Let $\mu$ be a probability measure on $\mathbb{R}^{d_1+d_2}$, with a density with respect to the Lebesgue measure. We assume there is a flow associated to a vector field $\Phi$ on $\mathbb{R}^{d_1+d_2}$ that leaves $\mu$ invariant, that is 
\[
    \dive(\mu \Phi) = 0.
\]
For the purpose of sampling $\mu$, following this flow would suffice. However, in many cases an analytical solution to this flow is unknown. 

We propose here to split the flow in parts that are solvable when handled separately. For any $z \in \mathbb{R}^{d_1+d_2}$, we write $z = (x,y)$ with $x\in \mathbb{R}^{d_1}$ and $y \in \mathbb{R}^{d_2}$.
We write the vector field $\Phi(z) = (\Phi_x(z),\Phi_y(z))$. Since $\Phi$ preserves $\mu$, we get at any given point $z$:
\begin{equation}
    \dive(\Phi(z) \mu(z)) = \dive_x(\Phi_x(z)\mu(z)) + \dive_y(\Phi_y(z)\mu(z)) = 0    
    \label{eq:full-flow-invariance}
\end{equation}
where $\dive_x$ and $\dive_y$ are the divergences defined on the respective spaces. We define the \textit{split pdmp} on the state space  $\mathbb{R}^{d_1+d_2} \times \{0,1\}$ as follow:
\begin{enumerate}
    \item deterministic dynamics follow the vector field $((1-\alpha)\,\Phi_x(z),\, \alpha\,\Phi_y(z))$ for $(z,\alpha) \in \mathbb{R}^{d_1+d_2} \times \{0,1\}$, so that depending on $\alpha$ the system evolves according to either the $x$-flow or the $y$-flow
    \item at each jump event, the kernel flips $\alpha$, switching which flow governs the dynamics
    \item if $\alpha = 0$, the event rate is 
    \[
        \lambda(z,0) = \Biggl[\frac{1}{\mu(z)}\,\operatorname{div}_x\!\bigl(\Phi_x(z)\,\mu(z)\bigr)\Biggr]^+,
    \] 
    \item if $\alpha = 1$, the event rate is 
    \[
        \lambda(z,1) = \Biggl[\frac{1}{\mu(z)}\,\operatorname{div}_y\!\bigl(\Phi_y(z)\,\mu(z)\bigr)\Biggr]^+.
    \] 
\end{enumerate}

\begin{remark}
    \label{remark:div-choice}
    From equation \ref{eq:full-flow-invariance}, we have 
    \[
    \dive_x(\Phi_x(z)\,\mu(z)) = -\,\dive_y(\Phi_y(z)\,\mu(z)),
    \] 
    so the event rate $\lambda$ can be computed using either divergence, depending on which is more convenient. Consequently we can derive the rates from a single scalar $\rho$ such that
    \[
    \lambda(z,0) = [\rho]^+ \qquad \qquad \lambda(z,1) = [-\rho]^+.
    \]
\end{remark}

\begin{prop}[Split PDMP]
   The measure $\nu = \mu \times \mathrm{Unif}(\{0,1\})$ is invariant by the split PDMP.
   \label{prop:split-pdmp-invariance}
\end{prop}

\begin{proof}
Let $(\nu_t)_{t \geq 0}$ denote the law of the process at time $t$, with initial distribution $\nu$. To verify invariance, it suffices to show that $\nu_t = \nu$ for all $t \geq 0$.  
Consider the mass balance at $(z,0)$:
\[
\frac{d \nu_t(z,0)}{dt} = \dive_x\bigl(\Phi_x(z)\,\nu_t(z,0)\bigr) - \lambda(z,0)\,\nu_t(z,0) + \lambda(z,1)\,\nu_t(z,1).
\]
Substituting $\nu_t = \nu$ and using the definitions of $\lambda$, we get
\[
\frac{d \nu_t(z,0)}{dt} = \dive_x\bigl(\Phi_x(z)\,\mu(z)\bigr) - \mu(z)\frac{1}{\mu(z)} \dive_x\bigl(\Phi_x(z)\,\mu(z)\bigr) = 0.
\]
A analogous argument holds for $(z,1)$, noting that 
\[
\dive_y(\Phi_y(z)\,\mu(z)) = -\,\dive_x(\Phi_x(z)\,\mu(z)).
\] 
Thus $\nu$ is invariant under the split PDMP.
\end{proof}

%%%%%%%%%%%%%%%%%%%%%%%%%%%%%%%%%%%%%%%%%%%%%%%%%%%%%%%%%%%%%%%%%%%%%%%%%%%
\subsection{Lagrangian dynamical Monte Carlo}
\label{subsec:lagragian-dynamical-mc}

\textit{Throughout this section and in what follows, we adopt Einstein notation. This notation clarifies computations and proofs. Results are stated in traditional (vector-matrix) notation whenever possible, so that the main arguments remain accessible. The detailed derivations require familiarity with Einstein notation, but implementing the proposed algorithms does not.}

Our aim is to sample a probability density $\pi$ for $x \in \mathbb{R}^d$. The core idea of the classical Hamiltonian Monte Carlo algorithm (that we shall call \textit{Euclidean HMC}) is to use a dynamical system described by a particular Hamiltonian that leaves $\pi$ \cite{HMCNeal} invariant. There, the target $\pi$ is extended to a phase space distribution for $(x,p) \in\R^{2d}$, where the conditional momentum is assumed to be normally distributed:
\[
p|x \sim \mathcal{N}(0, \Sigma). 
\]
The covariance matrix $\Sigma$, or equivalently, the mass matrix associated to the system is a crucial parameter for efficient sampling. Sampling highly correlated distributions requires $\Sigma$ to align with the covariance structure of the distribution $\pi$. One important limitation of Euclidean HMC is that $\Sigma$ is fixed, therefore any distribution exhibiting several relevant covariance structures, such as a mixture of gaussians, will be be difficult to sample.
To overcome this difficulty, \textit{Riemannian Hamiltonian Monte Carlo} considers more general covariance matrices where $\Sigma = G(x)$ \cite{girolami2011riemann}. The dynamical systems that result from positive definite $G(x)$ describe the motion of a point in a curved space with $G^{-1}$ as a metric. In other words, the setting is that of a Riemannian manifold which earns the corresponding sampling methods the name \textit{Riemannian manifold HMC}. 

In order to ease the numerical problems faced by \textit{Riemannian Hamiltonian Monte Carlo}, in Lagrangian Dynamical Monte Carlo \cite{Lan_2015} the authors reframe the Hamiltonian system described above into a Lagrangian system, i.e. a dynamical system using velocities $v$ instead of the momenta $p$.

\begin{remark}
    In Euclidean HMC, the momentum $p$ and velocity $v$ coincide if the mass matrix is identity, and the distinction is generally not important as $p = mv$. This is not the case in the Riemannian HMC setting, where $v = G(x)^{-1} p$.
\end{remark}

In Lagrangian Dynamical Monte Carlo, the velocity $v$ is  normally distributed conditionally on the position $x$:
    \[
        v | x\sim \mathcal{N}(0,G(x)^{-1}).
    \]
This yields a distribution $\mu$ for $(x,v) \in \R^{2d}$, which is left invariant by the Lagrangian of a perturbed geodesic. The flow is defined by the equations
\begin{align*}
    \dot{x} &= v \\
    \dot{v} &= -\eta(x,v) - G(x)^{-1} \nabla_x \bigl(- \log \pi(x) + \tfrac{1}{2} \log|G(x)|\bigr).
\end{align*}
where $\eta$ is a vector valued function that is quadratic in the velocity 
\[
(\eta(x,v))_k = \sum_{i,j} (\Gamma_{ij}(x))_k\, v_i v_j
\]
and where for each pair $i,j$ the vector $\Gamma_{ij}(x)$ has components
\[
(\Gamma_{ij}(x))_k = \tfrac{1}{2}\sum_l (G(x)^{-1})_{kl}\bigl(\partial_i G_{jl}(x)+\partial_j G_{il}(x)-\partial_l G_{ij}(x)\bigr).
\]
In this case, traditional notation is cumbersome, and the previous equations reads as follow in Einstein notation:
\begin{align*}
\frac{dx^a}{dt} &= v^a
\\
\frac{dv^a}{dt} &= -\eta^a - G^{ab}\partial_b(-\log(\pi) +\frac{1}{2} \log \det G ))  
\end{align*}
where 
\[
\eta^a = \Gamma^a_{bc}v^bv^c
\]
and $\Gamma^a_{bc}$ are the \textit{Christoffel symbols} of the metric $G$, that is to say,
\[
\Gamma^{a}_{bc} = \frac{1}{2}G^{ad}(G_{cd,b} + G_{bd,c} - G_{bc,d}).
\]

Since Lagrangian Dynamical Monte Carlo is a reparameterization of Riemannian HMC, the measure $\mu$ is invariant by the ODE flow. 
\begin{proposition}
    Let $\Phi$ denote the vector field defining the flow on $\mathbb{R}^{2d}$. Then
    \[
        \dive(\mu \Phi) = 0.
    \]
\end{proposition}
\begin{proof}
    For a detailed proof, see \cite{Lan_2015}.
\end{proof}
In practice, the flow in Lagrangian Dynamical Monte Carlo is approximated numerically and incorporated into a Metropolis--Hastings correction step.

\subsubsection{Choice of the Metric G}
In principle there is some freedom in choosing $G$ but as mentioned above the choice impacts the sampling efficiency. In \cite{betancourt2013general} the soft absolute of the Hessian for the logarithmic density is used
\[
G = \wr \mathcal{H}_{\log \pi} \wr_\alpha,
\]
where $\mathcal{H}_f$ refers to the Hessian of $f$ and $\wr K \wr_\alpha$ denotes the soft absolute of a matrix $K$ of hardness $\alpha$. Explicitly
\[
\wr K \wr_\alpha = (\exp(\alpha K) + \exp(-\alpha K)) \cdot  K \cdot (\exp(\alpha K)+\exp(-\alpha K))^{-1}.
\]
We use the same metric, but for brevity we shall in later sections let the hardness $\alpha$ be implicit. As a simple illustration of the resulting relationship between the covariance structure, the geometry and velocity distribution, see figure \ref{fig:velocity-spaces} below. 

\begin{figure}[hbt!]
    \centering
    % First subfigure (no individual caption)
    \begin{subfigure}{0.45\textwidth}
        \centering
        \includegraphics[width=\textwidth]{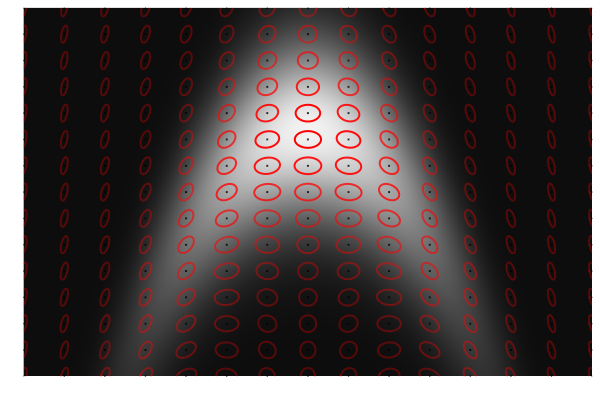}
    \end{subfigure}
    \hfill
    % Second subfigure (no individual caption)
    \begin{subfigure}{0.45\textwidth}
        \centering
        \includegraphics[width=\textwidth]{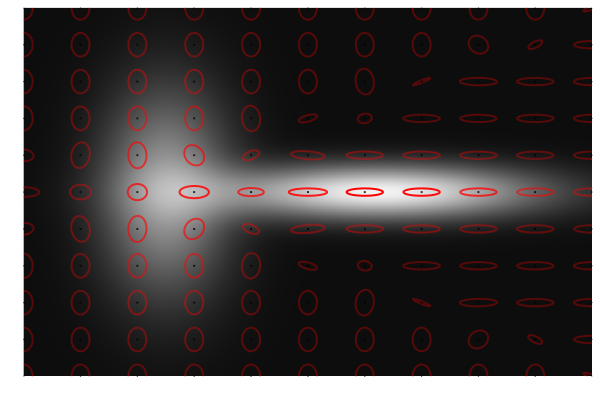}
    \end{subfigure}
    \caption{Velocity spaces for the banana distribution (left) and a mixture of gaussians (right), represented by level curves of the velocity distribution superimposed on the densities.}
    \label{fig:velocity-spaces}
\end{figure}

To construct the Christoffel symbols we shall need the derivatives of $G$. A convenient form of the derivative of $G$ is given in \cite{betancourt2013general}. Given the spectral decomposition 
\[
\mathcal{H}_{\log \pi} = Q D Q^{-1}
\]
the chosen metric becomes
\[
G = \wr \mathcal{H}_{\log \pi}\wr_\alpha = Q \wr D \wr_\alpha Q^T
\]
where the diagonal matrix $(\wr D \wr_\alpha)$ has elements
\[
(\wr D \wr_\alpha)_{ii} = \begin{cases}
    \lambda_i \coth (\alpha \lambda_i) \qquad &\text{if} \qquad \lambda_i \neq 0
    \\
    1/\alpha \qquad &\text{else}
    \end{cases}
\]
 Let $J$ be a matrix determined by the eigenvalues $\lambda_i = D_{ii}$ as follows 
\[
J_{ij} = \begin{cases}
    ({\lambda_i - \lambda_j})^{-1}(\lambda_i \coth{\alpha \lambda_i} - \lambda_j \coth{\alpha \lambda_j}) \qquad &\text{if} \qquad \lambda_i \neq \lambda_j
    \\
    \coth(\lambda_i \alpha) - \alpha \lambda (\sinh(\alpha \lambda_i))^{-2} &\text{if} \qquad \lambda_i = \lambda_j \neq 0
    \\
    0 \qquad &\text{else}
    \end{cases}
\]
Moreover for notational simplicity, for each $i$ $M_i$ be the matrix associated to a derivative of the Hessian in the $i$:th direction given by
\[
M_i = Q^T (\partial_i \mathcal{H}_{\log \pi}) Q
\]
Then, given $J$ and $M_i$, the derivatives of $G$ can be computed as
\[
\partial_i G_{jk} = \partial_i (\wr \mathcal{H}_{\log \pi} \wr_\alpha)_{jk}= (Q (J \circ M_i) Q^T)_{jk}
\]
where $\circ$ denotes the Hadamard product.

%%%%%%%%%%%%%%%%%%%%%%%%%%%%%%%%%%%%%%%%%%%%%%%%%%%%%%%%%%%%%%%%%%%%%%%%%%%
%%%%%%%%%%%%%%%%%%%%%%%%%%%%%%%%%%%%%%%%%%%%%%%%%%%%%%%%%%%%%%%%%%%%%%%%%%%
\section{Covariance-Adaptive PDMPs}
\label{sec:new-algo}

This section introduces two new algorithms. The first, Split Lagrangian PDMP, arises from a splitting of the Lagrangian dynamics introduced in the previous section and serves as an intermediate scheme. The second, Covariance-Adaptive BPS, adjusts the former (and extends BPS) by replacing costly events with the bounce events of BPS, yielding our final algorithm. 
We begin by presenting material common to both methods. 

Both algorithms are designed to sample from a target density $\pi$ on $\mathbb{R}^d$. For this purpose, $\pi$ is extended to a joint distribution $\mu$ on $\mathbb{R}^{2d}$ by introducing a velocity component, as described in Section~\ref{subsec:lagragian-dynamical-mc}:
\[
\mu(x, v) = \pi(x)\, \mathcal{N}\bigl(v \mid 0, G(x)^{-1}\bigr),
\]
where $G$ denotes the SoftAbs metric derived from the Hessian of the log-target, as proposed in \cite{betancourt2013general}. 
The logarithm of $\mu$ is given by
\[
\log \mu = \log \pi - \frac{1}{2} v^\top G v + \frac{1}{2} \log \det G - \frac{d}{2} \log(2\pi).
\]

\begin{remark}
The following algorithms and proofs do not depend on the chosen metric $G$ and applies to metrics beyond the SoftAbs metric.
\end{remark}

For both algorithms, we employ the flow-splitting technique introduced in Section~\ref{subsec:flow-splitting} rather than solving the full dynamical system numerically. The split flows for both methods take the general form
\begin{align}
    \Phi_x(x, v) &= v, & \text{for } \alpha=0, \nonumber\\ 
    \label{eq:eom}
    \Phi_v(x, v) &= -\eta(x, v) - G(x)^{-1} \nabla_x \phi(x), & \text{for } \alpha=1,
\end{align}
where $\phi$ depends on the specific PDMP under consideration but is independent of the velocity $v$.

We present two convenient expressions for the divergence of $\Phi_v$. The first is particularly useful when $x$ is fixed:
\begin{proposition} \label{prop:compact div}
    The velocity-divergence of $\Phi_v$ satisfies 
    \[
        -\frac{1}{\mu} \, \mathrm{div}_v (\Phi_v \mu)
    = 2\tr(\Gamma) \cdot v + \Phi_v^T  G  v,
    \]
    with $(\tr(\Gamma))_j = \sum_{i} (\Gamma_{ij})_i$. Equivalently in Einstein notation:
    \begin{align*}
    -\frac{1}{\mu}\dive_\textbf{v}(\Phi_\textbf{v}\mu) 
    &=2\Gamma^i_{ij}v^j  + \Phi^i_\textbf{v}G_{ij}v^j.
    \end{align*}
\end{proposition}
\begin{proof}
    \begin{align*}
    -\frac{1}{\mu}\dive_v(\Phi_v\mu) &=-\partial_{v^i}(\Phi_v^i) - \Phi^i_v\partial_{v^i}\log \mu  
    \\
    &=-\partial_{v^i}(-\Gamma^i_{jk}v^jv^k) +  \frac{1}{2}  \Phi^i_v\partial_{v^i}G_{kl}v^kv^l
    \\
    &=2\Gamma^i_{ij}v^j  + \Phi^i_vG_{ij}v^j
    \end{align*}
\end{proof}
When the position $x$ is fixed, both $\mathrm{tr}(\Gamma)$ and $G$ remain constant and can therefore be precomputed. When instead $v$ is fixed and $x$ varies, the expression can be reformulated using directional derivatives along $v$, which avoids computing the full derivative of the metric $G$. In this setting, we employ the second convenient form:
\begin{proposition}
    \label{prop:div-v-fixed}
    The velocity-divergence of $\Phi_v$ satisfies 
    \[
    -\frac{1}{\mu} \, \mathrm{div}_v (\Phi_v \mu)
    =  \mathrm{Tr}\Bigl(G^{-1} \, \frac{\partial G}{\partial v}\Bigr) 
    - \frac{1}{2} \left(v^\top \frac{\partial G}{\partial v} \, v \right)
    - \nabla_v \phi \cdot v,
    \] 
    where $\frac{\partial G}{\partial v}$ is the directional derivative of $G$ in direction $v$.
    In Einstein notation:
    \[
    -\frac{1}{\mu} \, \mathrm{div}_v (\Phi_v \mu)
    = (G^{ij}-\frac{1}{2}v^iv^j)G_{ij,k}v^k  - \phi_{,j}v^j.
    \]
\end{proposition}
\begin{proof}
Using proposition \ref{prop:compact div} and the explicit form of $\Phi$ in equation \eqref{eq:eom} we find
    \begin{align}
    -\frac{1}{\mu}\dive_\textbf{v}(\Phi_\textbf{v}\mu) &=2\Gamma^i_{ij}v^j  + \Phi^i_\textbf{v}G_{ij}v^j
    \nonumber
    \\
    &=2\Gamma^i_{ij}v^j  - \Gamma^i_{kl}v^kv^lG_{ij}v^j- G^{ik}\phi_{,k}G_{ij}v^j
    \nonumber
    \\
    &=G^{ij}G_{ij,k}v^k  - G_{ij}\Gamma^i_{kl}v^kv^lv^j- \phi_{,j}v^j
    \nonumber
    \\
    &=(G^{ij}-\frac{1}{2}v^iv^j)G_{ij,k}v^k  - \phi_{,j}v^j \label{eq:velocity divergence}
    \end{align}
\end{proof}

%%%%%%%%%%%%%%%%%%%%%%%%%%%%%%%%%%%%%%%%%%%%%%%%%%%%%%%%%%%%%%%%%%%%%%%%%%%
\subsection{Split Lagrangian PDMP (SL-PDMP)}
The Split Lagrangian PDMP (SL-PDMP) is obtained by applying the flow-splitting construction of Section~\ref{subsec:flow-splitting} to the Lagrangian dynamics introduced in Section~\ref{subsec:lagragian-dynamical-mc}. 
Using the notation introduced above, we set
\[
    \phi = -\log \pi + \tfrac{1}{2}\log \det G
\]
in equation~\eqref{eq:eom}. 
The resulting SL-PDMP involves a single type of event, corresponding to a flip of $\alpha$ between $0$ and $1$. 
The event rates are given by
\[
\lambda(x,v,0) = [-\rho]^+,
\qquad
\lambda(x,v,1) = [\rho]^+,
\]
where
\[
\rho = -\mu^{-1}\,\mathrm{div}_v(\Phi_v \mu).
\]
With these expressions established, the invariance of the process can be stated as follows.
\begin{proposition}
    The Lagrangian PDMP defined above preserves the measure $\mu \times \mathrm{Unif}(\{0,1\})$.
\end{proposition}
\begin{proof}
    This is a direct consequence of proposition \ref{prop:split-pdmp-invariance}.
\end{proof}
When the position is fixed, we instead use the expression for $\rho$ given in Proposition~\ref{prop:compact div}. When the velocity is fixed, we derive from proposition \ref{prop:div-v-fixed} and the expression of $\phi$ the following result for $\rho$:
\begin{lemma}\label{lemma:rho SL}
    $\rho$ can be written as
    \[
        \rho = \tfrac{1}{2}\,\mathrm{Tr}\!\left(G^{-1} \, \tfrac{\partial G}{\partial v}\right) 
        - \tfrac{1}{2}\, v^\top \Bigl(\tfrac{\partial G}{\partial v}\Bigr) v
        + \nabla \log \pi \cdot v,
    \]
    or, in Einstein notation,
    \[
        \rho = \tfrac{1}{2}(G^{ij} - v^i v^j) G_{ij,k} v^k + v^i (\log \pi)_{,i}.
    \]
\end{lemma}

\begin{proof}
Starting from Proposition~\ref{prop:div-v-fixed} 
\begin{align*}
    \rho &= \Bigl(G^{ij} - \tfrac{1}{2} v^i v^j\Bigr) G_{ij,k} v^k - \phi_{,j} v^j \\
         &= \Bigl(G^{ij} - \tfrac{1}{2} v^i v^j\Bigr) G_{ij,k} v^k - \tfrac{1}{2}(G^{ij} G_{ij,k}) v^k + v^i (\log \pi)_{,i} \\
         &= \tfrac{1}{2}(G^{ij} - v^i v^j) G_{ij,k} v^k + v^i (\log \pi)_{,i}.
\end{align*}
\end{proof}
\noindent

%%%%%%%%%%%%%%%%%%%%%%%%%%%%%%%%%%%%%%%%%%%%%%%%%%%%%%%%%%%%%%%%%%%%%%%%%%%
\subsection{Covariance-Adaptive BPS (CA-BPS)}

In the expression for $\rho$ derived in SL-PDMP, the term $\nabla \log \pi \cdot v$ corresponds exactly to the contribution appearing in the classical Bouncy Particle Sampler (BPS). 
This observation allows this part of the divergence of the flow $\Phi_x$ to be compensated by BPS-like events, which in turn eliminates the corresponding term in the flow $\Phi_v$. 

As in SL-PDMP, the state space is $\mathbb{R}^{2d} \times \{0,1\}$, but the dynamics are now governed by equation~\eqref{eq:eom} with
\[
\phi = \frac{1}{2}\log \det G.
\]
Compared with the SL-PDMP formulation, the $\log \pi$ term is removed from the expression for $\phi$. 
For this PDMP, two distinct types of events are defined, resulting in dynamics that are slightly more complex than in SL-PDMP. 
The first type, referred to as \emph{BPS events}, is defined by the rates
\begin{align*}
    \lambda_{\mathrm{BPS}}(x,v,0) &= [-v \cdot \nabla_x \log \pi(x)]^+, &
    \lambda_{\mathrm{BPS}}(x,v,1) &= 0,
\end{align*}
and acts on the state $(x,v,0)$ by reflecting the velocity along the gradient direction $w = \nabla_x \log \pi(x)$:
\[
    F: v \mapsto v - 2 \frac{\langle v, w \rangle_{G(x)}}{\| w \|^2_{G(x)}} w,
\]
where the inner product and norm are defined by the metric at $x$. 
The reflection operator $F$ satisfies the following properties: (i) it leaves the target density invariant, i.e., $\mu(x,F(v)) = \mu(x,v)$; (ii) it preserves volumes; and (iii) it ensures that
\[
\lambda_{\mathrm{BPS}}(x,F(v),0) = [-F(v) \cdot \nabla_x \log \pi(x)]^+ = [v \cdot \nabla_x \log \pi(x)]^+,
\]
which in turn implies that
\[
    \lambda_{\mathrm{BPS}}(x,F(v),0) - \lambda_{\mathrm{BPS}}(x,v,0) = v \cdot \nabla_x \log \pi(x).
\]
Thus, although the events are directly influenced by the metric, the event rate itself remains unchanged. In practice this means that this rate remains straightforward to compute.

In addition to BPS events, the PDMP includes \emph{Lagrangian events}, which flip the value of $\alpha$ between $0$ and $1$. 
The corresponding event rates are given by
\begin{align*}
    \lambda_L(x,v,0) &= [-\rho_L]^+, & \lambda_L(x,v,1) &= [\rho_L]^+,
\end{align*}
where
\[
    \rho_L = -\frac{1}{\mu}\,\mathrm{div}_v(\Phi_v \mu).
\]

\begin{proposition}
    The Covariance-Adaptive BPS (CA-BPS) preserves the measure $\nu = \mu \times \mathrm{Unif}(\{0,1\})$.
\end{proposition}
\begin{proof}
The proof can be viewed as a consequence of the fact that in the previous section, the quantity $\rho$ that was defined can be decomposed as
\[
    \rho = \rho_L + v \cdot \nabla_x \log \pi(x).
\] 
Nevertheless, a complete derivation is presented here, following the steps of proposition \ref{prop:split-pdmp-invariance}.
Let $(\nu_t)_{t \geq 0}$ denote the law of the process at time $t$, with initial distribution $\nu$.  
To verify invariance, it suffices to show that $\nu_t = \nu$ for all $t \geq 0$.  

\paragraph{Mass conservation for the $\mathbf{x}$-flow.} 
\[
\frac{\partial \nu_t(x,v,0)}{\partial t} = -\frac{1}{\mu}\, \mathrm{div}_x(\mu \Phi_x) + \lambda_{\mathrm{BPS}}(x,F(v),0) - \lambda_{\mathrm{BPS}}(x,v,0) + \lambda_L(x,v,1) - \lambda_L(x,v,0).
\] 
Expanding the terms gives
\begin{align*}
\frac{\partial \nu_t(x,v,0)}{\partial t}
&= -\mu^{-1} v \cdot \nabla_x \mu + [v \cdot \nabla_x \log \pi]^+ - [-v \cdot \nabla_x \log \pi]^+ + [\rho_L]^+ - [-\rho_L]^+ \\
&= -v^i \partial_i \log \mu + v^i \partial_i \log \pi + \tfrac{1}{2} (G^{ij} - v^i v^j) G_{ij,k} v^k.
\end{align*}
From the expression of $\log \mu$, one finds
\begin{align*}
- v^i \partial_i \log \mu 
&= -v^i \partial_i \log \pi + \tfrac{1}{2} G_{ij,k} v^i v^j v^k - \tfrac{1}{2} G^{ij} G_{ij,k} v^k \\
&= -v^i \partial_i \log \pi - \tfrac{1}{2} (G^{ij} - v^i v^j) G_{ij,k} v^k,
\end{align*}
so that
\[
\frac{\partial \nu_t(x,v,0)}{\partial t} = 0.
\]

\paragraph{Mass conservation for the $\mathbf{v}$-flow.} 
\[
\frac{\partial \nu_t(x,v,1)}{\partial t} = -\frac{1}{\mu}\, \mathrm{div}_v(\mu \Phi_v) + \lambda_L(x,v,0) - \lambda_L(x,v,1).
\] 
Expanding the terms and using proposition \ref{prop:div-v-fixed} yields
\begin{align*}
\frac{\partial \nu_t(x,v,1)}{\partial t}
&= (G^{ij} - \tfrac{1}{2} v^i v^j) G_{ij,k} v^k - \phi_{,j} v^j + [-\rho_L]^+ - [\rho_L]^+ \\
&= (G^{ij} - \tfrac{1}{2} v^i v^j) G_{ij,k} v^k - \tfrac{1}{2} G^{ij} G_{ij,k} v^k - \tfrac{1}{2} (G^{ij} - v^i v^j) G_{ij,k} v^k \\
&= 0.
\end{align*}
Hence, the “mass” is conserved for both the $x$ and $v$ flows, and therefore $\nu$ is invariant.

\end{proof}
As in the case of SL-PDMP, when the position is fixed (i.e., $\alpha = 1$), the expression of Proposition~\ref{prop:compact div} is used to compute $\rho_L$.  
When instead the velocity is fixed, the following lemma gives a convenient expression for the rates.
\begin{lemma}

    $\rho_L$ can be written as
    \[
        \rho = \tfrac{1}{2}\,\mathrm{Tr}\!\left(G^{-1} \, \tfrac{\partial G}{\partial v}\right) 
        - \tfrac{1}{2}\, v^\top \Bigl(\tfrac{\partial G}{\partial v}\Bigr) v
    \]
    or, in Einstein notation,
    \[
        \rho_L = \tfrac{1}{2}(G^{ij} - v^i v^j) G_{ij,k} v^k
    \]
\end{lemma}
\begin{proof}
    As for Lemma \ref{lemma:rho SL} we use Proposition \ref{prop:div-v-fixed}
\begin{align*}
    \rho_{L} &= -\frac{1}{\mu}\dive(\Phi_\textbf{v}\mu)
    \\
    &=(G^{ij}-\frac{1}{2}v^iv^j)G_{ij,k}v^k  - \phi_{,j}v^j 
    \\
    &= (G^{ij} -\frac{1}{2}v^iv^j)G_{ij,k}v^k-\frac{1}{2}G^{ij}G_{ij,k}v^k 
    \\
    &= \frac{1}{2}(G^{ij} - v^iv^j)G_{ij,k} v^k
\end{align*}
\end{proof}
\begin{remark}
    The PDMP can equivalently be formulated without including $\alpha$ in the state space, by instead treating each traversal along the velocity flow (corresponding to $\alpha = 1$) as an event. 
\end{remark}

%%%%%%%%%%%%%%%%%%%%%%%%%%%%%%%%%%%%%%%%%%%%%%%%%%%%%%%%%%%%%%%%%%%%%%%%%%%%%%%
%%%%%%%%%%%%%%%%%%%%%%%%%%%%%%%%%%%%%%%%%%%%%%%%%%%%%%%%%%%%%%%%%%%%%%%%%%%%%%%
%%%%%%%%%%%%%%%%%%%%%%%%%%%%%%%%%%%%%%%%%%%%%%%%%%%%%%%%%%%%%%%%%%%%%%%%%%%%%%%
%%%%%%%%%%%%%%%%%%%%%%%%%%%%%%%%%%%%%%%%%%%%%%%%%%%%%%%%%%%%%%%%%%%%%%%%%%%%%%%
%%%%%%%%%%%%%%%%%%%%%%%%%%%%%%%%%%%%%%%%%%%%%%%%%%%%%%%%%%%%%%%%%%%%%%%%%%%%%%%

%%%%%%%%%%%%%%%%%%%%%%%%%%%%%%%%%%%%%%%%%%%%%%%%%%%%%%%%%%%%%%%%%%%%%%%%%%%
%%%%%%%%%%%%%%%%%%%%%%%%%%%%%%%%%%%%%%%%%%%%%%%%%%%%%%%%%%%%%%%%%%%%%%%%%%%
\section{Simulation of the processes}
\label{sec:simulation}
As mentioned before, the traditional thinning procedure for simulating event times in PDMP samplers requires manual derivation of event rates for each target, which can be too cumbersome for a practical implementation. This limitation was overcome by Metropolised PDMP samplers \cite{chevallier2025practicalpdmpsamplingmetropolis} for which a brief introduction is given in section \ref{subsec:metropolised-PDMP}. Since both SL-PDMP and CA-BPS involve event rates that are too complex to handle analytically, we adopt the Metropolised approach for simulating events.

The Metropolised approach assumes that the deterministic flow can be solved analytically. In both SL-PDMP and CA-BPS, the flow in $x$ is simply $x_t = x_0 + v_0 t$, so for $\alpha = 0$ the rates can be approximated using a piecewise-constant scheme within the Metropolised framework. However, the flow in $v$ does not have a known analytical solution. We therefore adapt the Metropolised PDMP framework by numerically simulating both the ODE for $v$ and the associated event time. 
Simulating the first event of a non-homogeneous Poisson process with rate $\lambda(t)$ is equivalent to solving \[\int_0^t \lambda(s)\,ds = -\log(u),\] with $u \sim \text{Unif}(0,1)$. Thus, when $\alpha = 1$ (i.e., the process follows the velocity flow), the numerical flow $v_t$ is integrated by an ODE solver alongside the integral $\int_0^t \lambda(x_0,v_s,1)\,ds$, until it reaches the threshold $-\log(u)$. At the same time, we compute the integral of the rate of the reverse process. We assume the numerical integration is sufficiently accurate so that both the flow and the rate integrals can be treated as exact. These integrals are then used in eq.~\ref{eq:proba-path} to compute the path probability, and in eq.~\ref{eq:acceptance-proba} to determine the acceptance rate.

It may seem very challenging to numerically solve these equations to sufficient precision, but this is not the case: since $x$ is fixed during the flow, all $x$-dependent terms only need to be computed once. Recall that the ODE in $v$ is
\begin{align*}
    \dot{v} &= -\eta(x,v) - G(x)^{-1} \nabla_x \phi(x),
\end{align*}
with
\[
\phi(x) =
\begin{cases}
- \log \pi(x) + \tfrac{1}{2} \log|G(x)|, & \text{for SL-PDMP}, \\[0.5em]
\tfrac{1}{2} \log|G(x)|, & \text{for CA-BPS},
\end{cases}
\]
and $(\eta(x,v))_k = \sum_{i,j} (\Gamma_{ij})_k(x)\, v_i v_j$. Here, as was the case for the rates, the Christoffel symbols $(\Gamma_{ij})_k(x)$ and the gradient terms $G(x)^{-1} \nabla_x \phi(x)$ can be precomputed. Both ODEs then reduce to quadratic forms in $v$ with fixed coefficients, so the numerical solver only requires minimal computation per step.

\begin{remark}
    There is a parallel between the numerical approximations discussed above and the implicit solver used in Riemannian Manifold HMC, which also relies on numerical approximation and introduces an inherent, unquantified, bias. It should be noted that the implicit solver the RM-HMC requires solving an equation in both position an momentum, while our solver only require positions.
\end{remark}

Finally, to compute the acceptance rate we need to determine the change of volume $\psi$ associated to the reversal of a path in equation \ref{eq:acceptance-proba}. As mentioned in remark \ref{rmk:vol-change} $\psi$ is the change of volume arising from the flow. More precisely, for each path segment with fixed $\alpha$ the map $z_{t_i}\mapsto z_{t_{i+1}}$ adjusts the volume by a factor of
\[
\psi_i = \exp(\int_0^T \dive_a(\Phi_a)dt)
\]
where $a = x$ if $\alpha = 0$ and $a = v$ otherwise.

In our case the flow in $x$ is linear, so clearly $\dive_x(\Phi_x)=0$. The vector field defining the velocity flow (for either split PDMP given above) is quadratic in $v$, and we get a familiar term
\[
\dive_v(\Phi_v) = 2\Gamma^i_{ij}v^j = G^{ij}G_{ij,k}v^k,
\]
or in classical notation
\[
\dive_v(\Phi_v) = \tr (G^{-1}\frac{\partial G}{\partial v}).
\]
Thus along any segment where a velocity flow is followed it is necessary to compute the time integral of the above divergence. Such a term is already present in the rates, and thus this does not alter the difficulty in simulation.

%%%%%%%%%%%%%%%%%%%%%%%%%%%%%%%%%%%%%%%%%%%%%%%%%%%%%%%%%%%%%%%%%%%%%%%%%%%
%%%%%%%%%%%%%%%%%%%%%%%%%%%%%%%%%%%%%%%%%%%%%%%%%%%%%%%%%%%%%%%%%%%%%%%%%%%
\section{Experiments}
\label{sec:experiments}
The goal of the experiments in this section is to provide guidance on when our algorithm should be preferred over BPS. To this end, we consider two simple but informative posterior distributions. As a first benchmark, we use the two-dimensional banana distribution, a classical nonlinear target in the MCMC literature. We then investigate the role of spectral properties of the covariance structure of the posterior by considering $20$-dimensional Gaussian distributions with anisotropic covariance matrices. In particular, varying the spectral gap in the covariance structure provides insight into how the relative performance of the algorithms depends on the geometry of the target. Consequently, if practitioners have prior knowledge of the covariance structure of their target distribution, especially the spectral gap, they can use these experiments as a reference when selecting an appropriate sampling algorithm.  

To assess the efficiency of the two samplers on these benchmarks, we measure the Kolmogorov–Smirnov (KS) distance between the (known) cumulative distribution of the first-coordinate marginal of the target and the empirical distribution function obtained from the samples. The Kolmogorov–Smirnov distance between two cumulative distribution functions $F(x)$ and $G(x)$ is defined as
\[
D_{\mathrm{KS}}(F, G) = \sup_{x \in \mathbb{R}} \, \big| F(x) - G(x) \big|.
\]

Our method typically requires fewer events than standard BPS, but each event is more expensive to compute. In our experiments we deliberately use simple, cheap posteriors. To extrapolate to realistic (large-\(n\)) settings, where \(n\) denotes the number of data points, we adopt a simple cost model.

\paragraph{Computational cost model}
Our simulations only account for fixed costs, as they are performed on toy targets. If the target were more complex, introducing an additional computation time of $\epsilon$ per event for BPS, then CA-BPS would incur an extra cost of $\beta \epsilon$, where $\beta$ reflects the relative overhead of computing third-order derivatives compared to first-order ones. The costs reported in our toy experiments correspond to the special case $\epsilon = 0$.  

Let $A$ and $B$ denote the runtime costs of BPS and CA-BPS, respectively, when run for a duration yielding equivalent Kolmogorov--Smirnov distance across both algorithms. Let $k_{\mathrm{BPS}}$ and $k_{\mathrm{CA-BPS}}$ denote the corresponding number of events. Including the extra computation time per event $\epsilon$ gives overall runtimes of  
\[
A + k_{\mathrm{BPS}}\,\epsilon \quad \text{for BPS, and} \quad B + k_{\mathrm{CA-BPS}}\,\beta\,\epsilon \quad \text{for CA-BPS.}
\]  
Accordingly, we define the efficiency ratio as  
\begin{equation}
    \label{eq:ratio-def}
    r(\epsilon) \;=\; \frac{B + k_{\mathrm{CA-BPS}}\,\beta\,\epsilon}{A + k_{\mathrm{BPS}}\,\epsilon}.    
\end{equation}  
Note that $r(0) = B/A$ (the experimentally observed ratio), while  
\[
\lim_{\epsilon \to \infty} r(\epsilon) \;=\; \frac{k_{\mathrm{CA-BPS}}\,\beta}{k_{\mathrm{BPS}}},
\]  
so for large $\epsilon$, the comparison depends only on event counts and the relative overhead factor $\beta$.

In a Bayesian setting, the target distribution is a posterior distribution, whose complexity grows with the amount of available data. The parameter $\epsilon$ can be interpreted as the additional computational cost induced by extra data, while $\beta$ represents the relative per-data cost of CA-BPS compared to BPS.  

A brief note on the parameter \(\beta\). Although \(\beta\) can be observed empirically in experiments, any estimate obtained on our toy posteriors does not necessarily generalize to practical targets. For this reason we perform a sensitivity analysis over a range of \(\beta\) values rather than relying on a single empirical estimate. The case \(\beta=1\) is a natural and informative baseline: it describes the neutral regime in which per-datum costs in the bayesian setting are comparable for CA-BPS and BPS (so that differences in overall efficiency stem primarily from fixed overheads and iteration counts). This regime arises, for example, when sufficient statistics or precomputed aggregates remove most per-datum arithmetic (so both methods pay the same marginal cost for increasing the number of data), or when data-fetching and I/O dominate so that arithmetic differences are effectively masked. We present \(\beta=1\) but also additional scenarios with \(\beta > 1\) (corresponding to more expensive higher-order operations).

\begin{remark}
    Typically, algorithms rely only on first-order derivatives and incur small fixed costs, rendering this analysis unnecessary.
\end{remark}

\paragraph{Methodology}  
Each experiment is repeated 100 times to ensure robust performance estimates, and we report the median Kolmogorov--Smirnov (KS) error across these runs. Each algorithm is allowed to run for a fixed duration of 10 seconds per run. Before each set of runs, algorithm parameters are optimized using a nested Brent optimization over the refresh time $T$ and the adaptivity tolerance parameter. Both the inner and outer loops of the optimization use 10 steps, and the optimization criterion is the median KS error over 100 runs, consistent with the metric used for evaluation. We compare two algorithms in our experiments: standard BPS and CA-BPS.

%%%%%%%%%%%%%%%%%%%%%%%%%%%%%%%%%%%%%%%%%%%%%%%%%%%%%%%%%%%%%%%%%%%%%%%%%%%
\subsection{Banana-shaped posterior}

As a first illustrative example, we consider a two-dimensional ``banana'' distribution, a classical test case for sampling algorithms that exhibit strong non-linear correlations. The log-density is defined as
\[
\log \pi(x_1, x_2) = -b \,(x_2 - x_1^2)^2 - a\,(1 - x_1)^2,
\]
with parameters \(a = 1/20\) and \(b = 5000\). This distribution exhibits a curved, narrow valley, which makes it challenging for standard samplers. 
Results for $\beta = 1, 2, 10, 100 $ are shown in Figure \ref{fig:exp-banana-efficiency}. For $\beta = 1$, we observe that $r(0) < 1$, indicating that BPS is more efficient for our toy posterior. However, as the cost of posterior evaluation increases, CA-BPS eventually becomes more efficient, overtaking BPS for posteriors with evaluation times around $10^{-5}$\,s. For $\beta \geq 10$, BPS is always more efficient. This illustrates that the optimal choice of algorithm is highly context-dependent: CA-BPS should generally not be preferred for this target unless posterior evaluation incurs a substantial computational cost and $\beta$ is close to 1.

\begin{figure}
    \centering
    \includegraphics[width=0.99\textwidth]{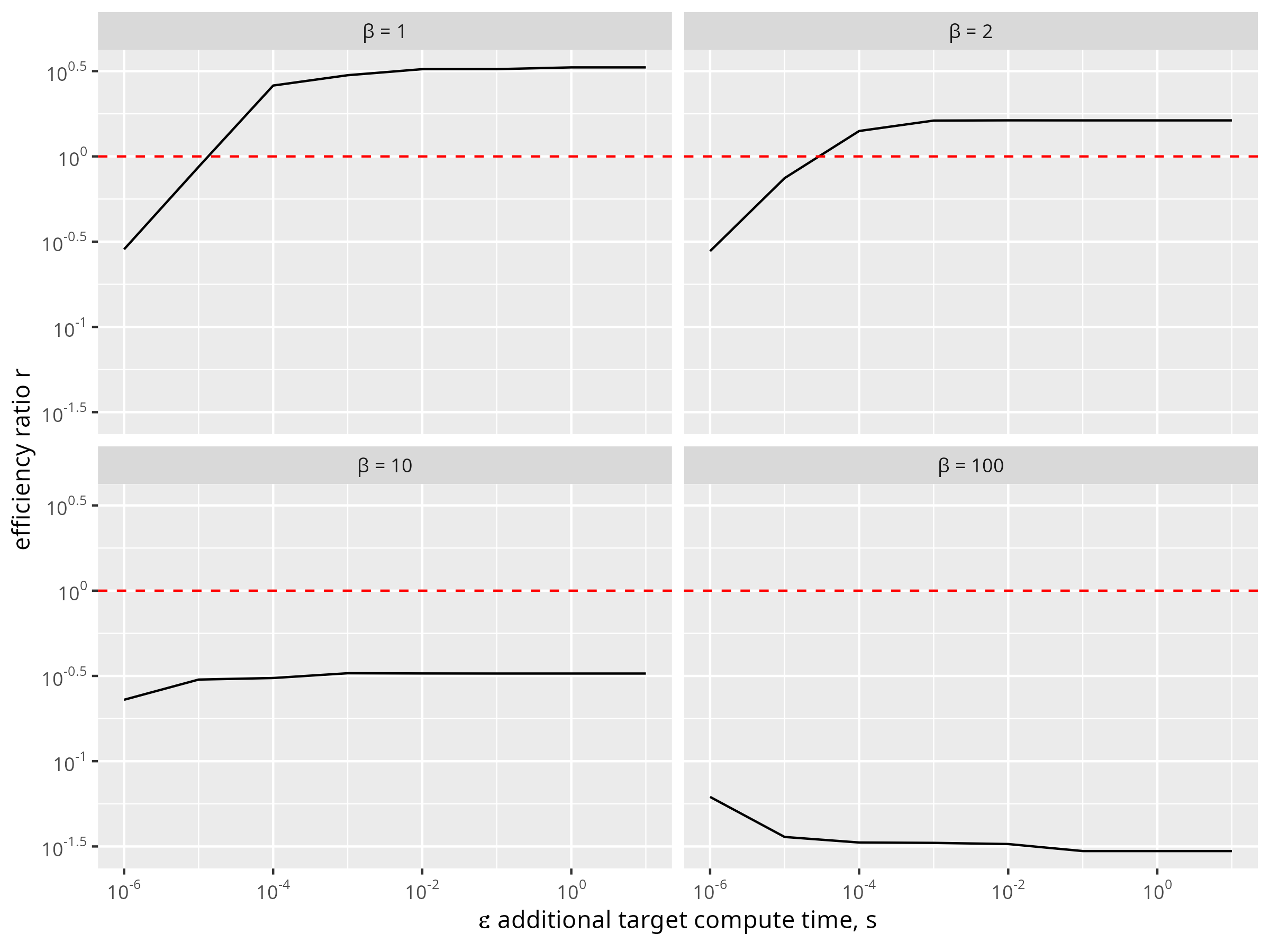}
    \caption{Efficiency ratio of CA-BPS relative to BPS for the banana-shaped posterior, across different values of $\beta$. Regions above the dashed red line indicate a relative advantage for CA-BPS; regions below indicate BPS is more efficient.}
    \label{fig:exp-banana-efficiency}
\end{figure}

%%%%%%%%%%%%%%%%%%%%%%%%%%%%%%%%%%%%%%%%%%%%%%%%%%%%%%%%%%%%%%%%%%%%%%%%%%%
\subsection{Gaussian target}

As a second benchmark, we consider a $20$-dimensional Gaussian distribution with mean zero and covariance matrix 
\[
\Sigma = \mathrm{Diag}(1, 1/\delta, \dots, 1/\delta),
\]
where $\delta$ controls the spectral gap. We vary $\delta$ across values $\delta = 10^i$ for $i = 1, \dots, 5$, thereby spanning regimes from nearly isotropic to strongly anisotropic covariance structures. This setup allows us to investigate how the efficiency of the algorithms depends on the spectral properties of the target distribution.

Figure~\ref{fig:exp-gaussian-20d-efficiency-all} shows that for $\epsilon = 0$, CA-BPS becomes more efficient once the spectral gap $\delta$ lies between $10^2$ and $10^3$. The efficiency ratio can grow substantially, indicating large gains for CA-BPS, especially in the baseline case $\beta = 1$. For very large values of $\beta$ (e.g.\ $\beta = 1000$), however, CA-BPS may become less efficient as $\epsilon$ increases, reflecting the additional cost of higher-order derivatives.

\begin{figure}
    \centering
    \includegraphics[width=0.99\textwidth]{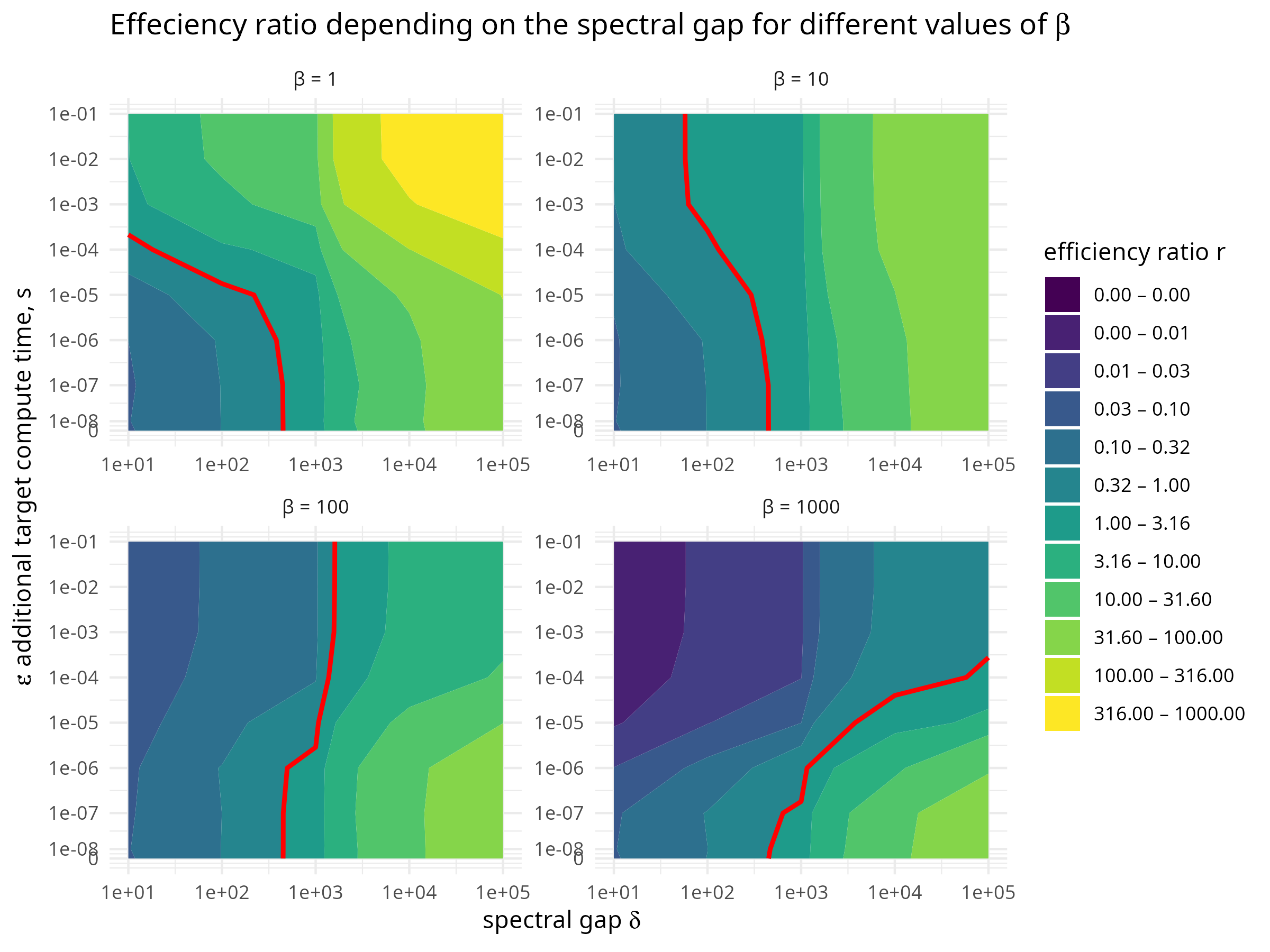}
    \caption{Efficiency ratio of CA-BPS relative to BPS for anisotropic Gaussian targets, as a function of spectral gap $\delta$ and $\beta$. The red line marks the neutral ratio of 1: values below indicate a relative advantage for BPS, while values above indicate CA-BPS is more efficient.}
    \label{fig:exp-gaussian-20d-efficiency-all}
\end{figure}

%%%%%%%%%%%%%%%%%%%%%%%%%%%%%%%%%%%%%%%%%%%%%%%%%%%%%%%%%%%%%%%%%%%%%%%%%%%
%%%%%%%%%%%%%%%%%%%%%%%%%%%%%%%%%%%%%%%%%%%%%%%%%%%%%%%%%%%%%%%%%%%%%%%%%%%
\section{Discussion}
\label{sec:discussion}

Building on the ideas of \cite{girolami2011riemann} and \cite{Lan_2015}, we developed a novel PDMP that incorporates local covariance information into its dynamics. The resulting process is numerically approximated within the Metropolised PDMP framework \cite{chevallier2025practicalpdmpsamplingmetropolis}.

The implementation of the proposed algorithms highlights several practical considerations and directions for improvement.  
First, to be applicable in generic cases automatic differentiation is crucial. Current libraries provide highly optimized routines for gradients, but support for second-order derivatives is less developed, and third-order derivatives are generally inefficient, leading to unnecessary memory allocations. The performance results reported here should therefore be interpreted as somewhat pessimistic with respect to what could be achieved with more mature tools.  

Second, since the covariance structure need not be exact, it may be sufficient in practice to approximate the Hessian using only a subset of the data. This could substantially reduce computational costs.  

Third, the algorithm is most promising in regimes where the spectral gap is relatively wide, which ensures efficient exploration. A particularly appealing application lies in multi-modal sampling, for instance mixtures of Gaussians with different correlation structures. When combined with techniques such as parallel tempering, the method could open new possibilities for tackling such problems.  

Looking forward, several avenues for improvement stand out. One is to develop efficient approximations by computing only part of the Hessian, thereby further reducing computational overhead. Another is to extend CA-BPS by treating the metric itself as an auxiliary state, updated only when it deviates significantly from its local value. Such a modification would avoid the need to recompute the metric at every event and could therefore yield substantial performance gains. Finally, while Metropolised PDMPs can incorporate the No-U-Turn criterion, we did not implement this here, as it is not yet clear how to define a “U-turn” in the presence of a position-dependent metric. Developing a suitable generalization of this idea constitutes an important direction for future work.

\printbibliography

\end{document}